\documentclass[conference]{IEEEtran}
\IEEEoverridecommandlockouts
\usepackage{float}
\usepackage{cite}
\usepackage{placeins}
\usepackage{amsmath,amssymb,amsfonts}
\usepackage{algorithmic}
\usepackage{graphicx}
\usepackage{textcomp}
\usepackage{xcolor}
\usepackage{mathtools}

\usepackage{tikz}
\usetikzlibrary {arrows.meta,bending,positioning}
\usepackage{amsfonts}
\usepackage{amsmath}
\tikzset{point/.style={fill,circle,inner sep=1pt}}
\tikzset{arr/.style={arrows = {-Latex[width=3pt, length=5pt]}}}
\tikzset{arr2/.style={arrows = {Latex-Latex[width=3pt, length=5pt]}}}

\usetikzlibrary{
	calc,
	patterns,
	positioning,
	angles
}
\usepackage{amsmath}

\DeclareMathOperator*{\argmin}{arg\,min}

\usepackage{hyperref}
\hypersetup{
    colorlinks=true,
    linkcolor=blue,
    filecolor=green,      
    urlcolor=cyan,
    citecolor=magenta,
    pdftitle={Optimal Clifford States for Transverse Ising Hamiltonians},
    pdfpagemode=FullScreen,
    }
\usepackage{amsthm}
\newtheorem{theorem}{Theorem}[section]

\newtheorem{corollary}[theorem]{Corollary}
\newtheorem{lemma}[theorem]{Lemma}

\newtheorem{definition}[theorem]{Definition}

\usepackage{array}
\usepackage{amsmath}

\def\BibTeX{{\rm B\kern-.05em{\sc i\kern-.025em b}\kern-.08em
    T\kern-.1667em\lower.7ex\hbox{E}\kern-.125emX}}

\begin{document}

\title{Optimal Clifford Initial States for \\Ising Hamiltonians
}
\author{\IEEEauthorblockN{ Bikrant Bhattacharyya}
\IEEEauthorblockA{\textit{Illinois Mathematics and Science Academy} \\
bbhattacharyya@imsa.edu}
\and
\IEEEauthorblockN{Gokul Subramanian Ravi
}
\IEEEauthorblockA{\textit{University of Michigan}\\
gsravi@umich.edu}

}

\maketitle

\thispagestyle{plain}
\pagestyle{plain}

\begin{abstract}
Modern day quantum devices are in the NISQ era, meaning that the effects of size restrictions and noise are essential considerations when developing practical quantum algorithms. Recent developments  have demonstrated that \textbf{V}ariational \textbf{Q}uantum \textbf{A}lgorithms (VQAs) are an appropriate choice for current era quantum devices. VQAs make use of classical computation to iteratively optimize parameters of a quantum circuit, referred to as the ansatz. These parameters are usually chosen such that a given objective function is minimized. Generally, the cost function to be minimized is computed on a quantum device, and then the parameters are updated using a classical optimizer. One of the most promising VQAs for practical hardware is the \textbf{V}ariational \textbf{Q}uantum \textbf{E}igensolver (VQE). VQE uses the expectation value of a given Hamiltonian as the cost function, and the minimal value of this particular cost function corresponds to the minimal eigenvalue of the Hamiltonian.

Because evaluating quantum circuits is currently very noisy, developing classical bootstraps that help minimize the number of times a given circuit has to be evaluated is a powerful technique for improving the practicality of VQE. One possible such bootstrapping method is creating an ansatz which can be efficiently simulated on classical computers for restricted parameter values. Once the optimal set of restricted parameters is determined, they can be used as initial parameters for a VQE optimization which has access to the full parameter space. Stabilizer states are states which are generated by a particular group of operators called the Clifford Group. Because of the underlying structure of these operators, circuits consisting of only Clifford operators can be simulated efficiently on classical computers. \textbf{C}lifford \textbf{A}nsatz \textbf{F}or \textbf{Q}uantum \textbf{A}lgorithms (CAFQA) is a proposed classical bootstrap for VQAs that uses ansatzes which reduce to clifford operators for restricted parameter values \cite{ravi2022cafqa}.

CAFQA has been shown to produce fairly accurate initialization for VQE applied to molecular Hamiltonians. Motivated by this result, in this paper we seek to analyze the Clifford states that optimize the cost function for a new type of Hamiltonian, namely Transverse Field Ising Hamiltonians. Our primary result is contained in theorem \ref{Main Result} which connects the problem of finding the optimal CAFQA initialization to a submodular minimization problem which in turn can be solved in polynomial time.
\end{abstract}

\section{Introduction}
Quantum computing (QC) represents a groundbreaking computational paradigm for solving specific problems that are traditionally impractical to tackle classically. 
It is anticipated that quantum computers (QCs) will wield significant advantages in areas of critical impact such as cryptography~\cite{Shor_1997}, chemistry~\cite{kandala2017hardware}, optimization~\cite{moll2018quantum}, and machine learning~\cite{biamonte2017quantum}.

During the ongoing Noisy Intermediate-Scale Quantum (NISQ) era, we are set to operate with quantum machines equipped with hundreds to thousands of imperfect qubits~\cite{preskill2018quantum}. 
In this era, these machines will grapple with limited connectivity and relatively short qubit lifetimes due to design constraints. 
Noise remains a major obstacle, preventing current quantum computers from outperforming classical computers in nearly all applications. 
Machines in the NISQ era will be incapable of executing extensive quantum algorithms like Shor Factoring~\cite{Shor_1997} and Grover Search~\cite{Grover96afast}. 
These algorithms necessitate error correction involving millions of qubits to establish fault-tolerant quantum systems~\cite{O_Gorman_2017}. 
However, a range of error mitigation approaches~\cite{czarnik2020error,Rosenberg2021,barron2020measurement,botelho2021error,wang2021error,temme2017error,li2017efficient,giurgica2020digital} have been proposed, enhancing execution fidelity on today's quantum devices. 
Nonetheless, the achieved fidelity still falls short of the requirements for most practical applications.

In recent times, there has been a growing focus on leveraging classical computing support to elevate the applicability of NISQ applications and devices in the real world. This effort encompasses various enhancements, such as optimizations at the compiler level~\cite{ravi2021vaqem}, advancements in classical optimizers~\cite{9259985}, circuit segmentation with classical compensation~\cite{CutQC,zhang2021variational}, among others. We are currently in the early stages of exploring this synergistic quantum-classical paradigm. There exists significant potential for employing sophisticated classical bootstrapping tailored to specific applications, thereby pushing the boundaries of NISQ capabilities forward. 

Variational quantum algorithms (VQAs) are anticipated to align well with NISQ machines, demonstrating a broad array of applications, including electronic energy estimation for molecules~\cite{peruzzo2014variational} and approximations for MAXCUT~\cite{moll2018quantum}. The quantum circuit in a VQA is defined by a set of angles, which are fine-tuned by a classical optimizer across multiple iterations to reach a specific target objective representing the VQA problem. These algorithms exhibit greater suitability for present-day quantum devices due to their ability to adapt to the idiosyncrasies and noise profile of the quantum machine~\cite{peruzzo2014variational, mcclean2016theory}. Regrettably, the accuracy achieved by VQAs on existing NISQ machines, even with error mitigation strategies, frequently falls significantly short of the exacting accuracy demands, particularly in domains like molecular chemistry, especially when dealing with larger problem sizes~\cite{wang2021error, kandala2017hardware, ravi2021vaqem}.

To advance NISQ VQAs towards practical utility, it is crucial to carefully select the parameterized circuit (ansatz) for a VQA and optimize its initial parameters classically to be as close to optimal as possible before venturing into quantum exploration. This approach holds the promise of enhancing accuracy and expediting algorithmic convergence on the quantum device, even in the presence of noise~\cite{mcclean2018,wang2020noise}.
While certain applications may derive advantages from domain-specific knowledge guiding the choice of particular parameterized circuits and initial parameters (e.g., UCCSD~\cite{romero2018strategies}), these choices are less appropriate for execution on contemporary quantum devices due to their substantial quantum circuit depth. Ansatz circuits tailored for today's devices, often termed as ``hardware efficient ansatz"~\cite{kandala2017hardware}, are typically agnostic to specific applications and stand to benefit significantly from judiciously chosen initial parameters. However, accurately estimating these parameters through classical means can be a challenging task.

Prominent prior work like CAFQA~\cite{ravi2022cafqa} focuses on initializing the VQA ansatz through classical simulation. In CAFQA, the initial parameters for VQA are selected through an efficient and scalable search within a classically simulable segment of the quantum space known as the Clifford space, employing Bayesian Optimization. CAFQA attains remarkable accuracy during initialization. Specifically, for the crucial chemistry application involving the estimation of the ground state energy of molecules, CAFQA restores up to 99.99\% of the accuracy lost in previous state-of-the-art classical initialization, demonstrating mean improvements of 56x.

Motivated by this result, in this paper we seek to analyze the Clifford states that optimize the cost function for a new type of Hamiltonian, namely Transverse Field Ising Hamiltonians. Our primary result is contained in theorem \ref{Main Result} which connects the problem of finding the optimal Clifford initialization to a submodular minimization problem which in turn can be solved in polynomial time. This connection arises by mapping the problem of finding the Hamiltonian ground state to a graph theoretic optimization problem, which is submodular. Submodular functions satisfy a criterion formalized in Appendix \ref{sub}. As a result of this property, submodular functions can be minimized by considering related constrained convex optimization problems \cite{chakrabarty2014provable}. 
When running numerical experiments, we consistently obtain Clifford approximations of the ground state energies that have a relative error varying from roughly 0-25\% when compared to the exact ground state energies.
This error itself is problem dependent and is an indicator of the true `quantum-ness' of the problem.
Once the good Clifford initialization is found, VQE can be run on future quantum devices (with reasonably low error rates), to find (nearly) exact ground state energy estimates.

\section{Background}

Here, we provide a brief overview of VQE. For a more in-depth overview of VQE see \cite{Tilly_2022}.
The \textit{variational principle} from quantum mechanics is a statement regarding the expectation value of a Hamiltonian. Given a Hamiltonian $H$ with ground state $E_0$, any quantum state $|\psi \rangle$ will satisfy 
\begin{equation}
E_0\leq \langle \psi | H | \psi \rangle 
\end{equation}
Thus if the state can be parameterized by a vector parameter angles $\theta\in[0,2\pi)^k$, then 
\begin{equation}
E_0 \approx \min_{\theta \in \mathbb{R}^k} \langle \psi(\theta) | H | \psi(\theta) \rangle
\end{equation}

Assuming that the parameterization of $|\psi\rangle$ is expressive enough to accurately predict the ground state of the Hamiltonian.
 
A common way to parameterize $|\psi\rangle$ is by using a parameterized quantum circuit which implements some parameterized unitary operator $U(\theta)$ acting on the initial state $|0\rangle ^{\otimes n}$ where $n$ is the number of qubits.\footnote{From here on we use the shorthand $|0\rangle^{\otimes n}=|0\rangle$. Generally, context will make it clear how many qubits we are considering.} Using this parameterization,
\begin{equation}
    |\psi\rangle = U(\theta)|0\rangle
\end{equation}
and
\begin{equation}
\langle \psi | H | \psi \rangle = \langle 0 |U^\dagger (\theta)HU(\theta)|0\rangle\end{equation}
The circuit which implements the unitary operator $U$ is called the \textit{ansatz}, and the choice of \textit{ansatz} is critial to the performance of VQAs \cite{Du_2022}.
 
The \textit{Variational Quantum Eigensolver} (VQE) begins with some set of parameters $\theta_i$ and some ansatz $U(\theta)$. It then repeatedly computes the expectation value of the Hamiltonian on $U(\theta)|0\rangle$ using a quantum circuit derived from the ansatz. Between each computation, the parameters are updated using a classical optimizer. The benefit of using quantum devices to find expectation values is that the gate depth of such circuits (both the gates required to construct the ansatz and those required to compute the expectation value once the ansatz has been applied) grows polynomially with $n$ whereas using purely classical simulation will require matricies of size $2^n\times 2^n$ at minimum.

CAFQA's primary proposal is to make use of an ansatz which lends itself to a classical initial search before requiring the usage of a quantum computer \cite{ravi2022cafqa}. It has been shown that circuits which are composed entirely of \textit{Clifford Gates} can be simulated in polynomial time using the \textit{stabilizer technique} \cite{Aaronson_2004}\cite{gottesman1998heisenberg}. The states obtained from applying Clifford Gates are called \textit{Stabilizer states} or \textit{Clifford states}.\footnote{We will use these two terms interchangeably.}
 
From here on, we will assume that the ansatz under consideration will be able to reach every $n$-qubit clifford state. With this assumption, determining the best CAFQA initial value is the same as determining the best clifford state initial value.
 
From here on
\begin{itemize}
    \item $|\psi\rangle$ will represent an arbitrary state
    \item $|\varphi\rangle$ will represent an arbitrary Clifford state.
    \item $|\psi_0\rangle$ will be the ground state for a given Hamiltonian $H$
    \item $|\varphi_0\rangle$ will represent the Clifford state that minimizes $\langle \varphi | H | \varphi \rangle$ over all Clifford states $|\varphi\rangle$.
\end{itemize}
\section{Mathematical Prerequisites}
Before moving onto determining the optimal clifford states for a given Hamiltonian, we require the development of some prerequisite theory.
\subsection{Graph Theory}
In this section we review the required graph theory concepts the notations that we will use. For a fuller overview of graph theory alongside proofs of some of the corollaries left unproven here, see \cite{graphTheoryRef}.
\begin{definition}\label{GraphDefinition}\normalfont
A \textit{graph} $G$ is given by a set of vertices $V_G$ and a set of edges $E_G$. We write this as $G=(V_G,E_G)$. \footnote{Graphs will always be denoted with capital letters.}
\end{definition}
The particular graphs that we will use for computations later on are defined in Appendix \ref{graph}.For our purposes, $V_G=\{q_0,q_1,q_2,\dots, q_{N-1}\}$ with $N\geq 2$. This labelling provides a natural correspondence between node $q_i$ and qubit $q_i$. An edge between nodes $q_i$ and $q_j$ will be denoted as $\langle q_i,q_j\rangle =\langle q_j,q_i\rangle$ \footnote{In general graph theory it's often useful to assign a real number weight to each edge. Furthermore, there are times where $\langle q_i,q_j\rangle \neq \langle q_j,q_i\rangle$ is a more natural choice. However, here we consider unweighted and undirected graphs, neglecting these alternate definitions.}. Furthermore, we only consider graphs with edges between distinct nodes.
\begin{definition}\label{SubgraphDefinition}\normalfont
A \textit{subgraph} $S=(V_S,E_S)$ of a given graph $G$ is a new graph satisfying $V_S\subset V_G$ and $E_S\subset V_E$. We write $S\subset G$ to say that $S$ is a subgraph of $G$. 
\end{definition}
There's a special type of subgraph that will play a special role later on.
\begin{definition}\label{induced}\normalfont
Given a set of nodes $\{q_{i_1},q_{i_2},\dots\}$ of a graph $G$, the subgraph $S$ \textit{induced} by $\{q_{i_1},q_{i_2},\dots\}$ is the subgraph with $V_S=\{q_{i_1},q_{i_2},\dots\}$ and $E_S$ equal to the set of all edges in $V_G$ between nodes in $\{q_{i_1},q_{i_2},\dots\}$. We will write $E(S)$ to denote the edges induced by vertex set $S$.
\end{definition}
Given a graph $G$ and two nodes $q_i$ and $q_j$, we can define whether or not $q_i$ or $q_j$ are connected as follows.
\begin{definition}\normalfont\label{connected}
Nodes $q_i$ and $q_j$ are $G$-\textit{connected} if either of the two following statements are true
\begin{itemize}
    \item $\langle q_i, q_j \rangle \in E_G$
    \item There exists a sequence of nodes $q_{k_1}, q_{k_2},\dots,q_{k_m}$ such that 
    $$\langle q_i, q_{k_1}\rangle,\langle q_{k_1},q_{k_2}\rangle,\dots, \langle q_{k_m}, q_j\rangle\in E_G$$
\end{itemize}
A graph $G$ is called \textit{connected} if every element of $V_G$ is $G$-connected to every other element in $V_G$.
\end{definition}
\begin{corollary}\normalfont
The following two immediately follow from Definition \ref{connected} and Definition.
\begin{itemize}
    \item $q_i$ being $G$-connected to $q_j$ is equivalent to $q_j$ being $G$-connected to $q_i$ (reflexivity).
    \item If $q_i$ is $G$-connected to $q_j$ and $q_j$ is $G$-connected to $q_k$ then $q_i$ is $G$-connected to $q_k$ (transitivity).
\end{itemize}
\end{corollary}
\begin{corollary}\normalfont
Every graph $G$ is the union of $k\geq 1$ disjoint connected subgraphs called \textit{connected components}. In other words, there exists a set of graphs $C_1,C_2,\dots, C_{k}\subset G$ such that 
\begin{itemize}
    \item $V_G=\bigcup_{i=1}^{k} V_{C_i}$
    \item $E_G=\bigcup_{i=1}^{k} E_{C_i}$
    \item Each $C_i$ is connected 
    \item $V_{C_i}\cap V_{C_j}=\emptyset$ if $i\neq j$
\end{itemize}
\end{corollary}
Connectivity allows us to define spanning trees as follows.

\begin{definition}\label{spanning tree}\normalfont
Given a connected graph $S$, we call any connected subgraph $M\subset S$ a \textit{spanning tree} of $S$ if $M$ is connected and satisfies the following properties.
\begin{itemize}
    \item $V_M = V_S$
    \item If $\langle q_i,q_j\rangle \in E_M$, then there exists no sequence $q_{k_1},q_{k_2},q_{k_m}$ with 
    $$\langle q_i, q_{k_1}\rangle,\langle q_{k_1},q_{k_2}\rangle,\dots, \langle q_{k_m}, q_j\rangle\in E_G$$
\end{itemize}
If $S$ is disconnected but the union of disjoint connected subgraphs $C_i$, then we define a spanning tree on $S$ as
\begin{equation}
M=\left(V_S,\bigcup _ i E_{M_i}\right)\end{equation}
Where $M_i$ is a spanning tree for $C_i$.\footnote{Sometimes spanning trees for disconnected graphs are called spanning forests. Here we will not distinguish them.}
\end{definition}
The following facts about spanning trees are well known \cite{scheffler2023graph}.
\begin{lemma}\label{bfs}
If $M$ is a spanning tree of a graph $S$ with $k$ connected components, then $|E_M|=|V_S|-k$. Furthermore, every graph $S$ contains at least one spanning tree $M$.
 
Given a graph $S=(V_S,E_S)$, there is an algorithm that runs in $O(|V_S|+|E_S|)$ that can always find the spanning tree for $S$.
 \end{lemma}

The following function will be a useful theoretical tool later on.
\begin{definition}\label{edge}\normalfont
Given a graph $G$, let the sequence of sets $\mathcal{Q}_n$ for $2\leq n\leq N$ be defined as\footnote{We use $P(S)$ to denote the power set of $S$.}
\begin{equation}
    \mathcal{Q}_n=\{S\in P(V_G)\textrm{  such that  }|S|\leq n\}\end{equation}
We define the following function as the \textit{edge-function}
\begin{equation}
    \mathcal{E}(n)=\max_{S\in \mathcal{Q}_n}|E(S)|
\end{equation}
For completeness we define $\mathcal{E}(0)=\mathcal{E}(1)=0$.
 
If $S$ is a set of vertices that maximizes $|E(S)|$ over $\mathcal{Q}_n$, we can $S$ an $n$-optimal vertex set.
 
Notice that for fixed $n$, there may be multiple $n$-optimal vertex sets. For completeness, we make the empty set the $0$-optimal vertex set and $\{q_k\}$ for $0\leq k < N$ the $1$-optimal vertex sets.
\end{definition}

\subsection{Stabilizer States}
The result of a Clifford Circuit will always be a stabilizer state, which can be equivalently defined as below:
\begin{definition}\normalfont
A $n$-qubit \textit{stabilizer state} $|\varphi\rangle$ is a state which has a stabilizer group consisting of $2^n$ elements of the $n$-qubit Pauli group.
 
In other words, every operator $O$ satisfying $O|\varphi\rangle=|\varphi\rangle$ is an element of the $n$-qubit Pauli Group and there are $2^n$ such operators.
\end{definition}\noindent
It's well known that the stabilizer group for state $|\varphi\rangle$, which will be denoted by $\mathcal{S}_\varphi$ can always be generated by $n$ operators as described in the Theorem below.
\begin{theorem}\label{Stabilizer Size}\normalfont
For any stabilizer state $|\varphi\rangle$, $\mathcal{S}_\varphi$ can be written as $\langle  \mathcal{G}_\varphi \rangle$ where $\mathcal{G}_\varphi=\{O_1,O_2,\dots, O_n\}$ consists of $n$ elements of the $n$-qubit Pauli group, each of which satisfies
\begin{itemize}
    \item  $O_i\notin \langle G_\varphi\backslash O_i \rangle$. 
    \item $[O_i,O_j]=0$ if $i\neq j$. 
\end{itemize}
This theorem has the immediate corollary
\begin{corollary}\label{allowed}\normalfont
Given qubit indices $i,j,k$ with $i\neq j$, $\pm Z_iZ_j$ and $\pm X_k$ cannot simultaneously be in $\mathcal{G}_\varphi$ for any state $|\varphi\rangle$ if $k\in\{i,j\}$.
 
Given distinct indices $l$ and $s$, no sequence of operators of the form $Z_{l}Z_{k_1},Z_{k_1}Z_{k_2},Z_{k_2}Z_{k_3},\dots Z_{k_m}Z_s$ can be in $\mathcal{G}_\varphi$ if $Z_lZ_s$ is and vice versa.\footnote{Notice the edges in a spanning tree follow a similar requirement.}
\end{corollary}
\end{theorem}\normalfont\noindent
The other well known result that will be used is the following.
\begin{theorem}\label{Stabilizer Expectation}\normalfont
Given an operator $O$ in the $n$-qubit Pauli group and a $n$-qubit stabilizer state $|\psi\rangle$, the expectation value $\langle \psi | O | \psi\rangle$ is always $0$, $\pm 1$, or $\pm i$. Furthermore, if $O$ is Hermitian, then the only possible expectation values are $0$ and $\pm 1$.
\end{theorem}\noindent
Proofs of Theorems \ref{Stabilizer Size} and \ref{Stabilizer Expectation} can be found in \cite{Aaronson_2004}.
The following lemma is inspired by \ref{Stabilizer Expectation}.
\begin{lemma}\label{ZX Sametime}
If $|\chi\rangle$ is stabilized by $X_j$ and $|\varphi\rangle$ is stabilized by $Z_iZ_j$, then 
\begin{itemize}
    \item $\langle \chi | Z_iZ_j|\chi\rangle =0$
    \item $\langle \varphi | X_j|\varphi\rangle =0$
\end{itemize}
\end{lemma}
\begin{proof}
The eigenspace of $X_j$ is spanned by states of the form $$|b_0b_1b_2\dots b_{j-1}\rangle\otimes |+\rangle\otimes |b_{j+1}b_{j+2}\dots b_{N}\rangle$$
Applying $Z_{i}Z_j$ results in one of the following states depending on whether $i>j$ or $i<j$.
\begin{itemize}
    \item $|b_0b_1\dots b_{i-1}\rangle\otimes |(1-b_i)\rangle \otimes |b_{i+1}b_{i+2}\dots b_{j-1}\rangle\otimes |-\rangle\otimes |b_{j+1}b_{j+2}\dots b_N\rangle $ if $i<j$
     \item $|b_0b_1\dots b_{j-1}\rangle\otimes |-\rangle\otimes |b_{j+1}b_{j+2}\dots b_{i-1}\rangle\otimes  |(1-b_i)\rangle \otimes |b_{i+1}b_{i+2}\dots b_N\rangle $ if $i>j$
\end{itemize}
It immediately follows that for every basis  state of the $+1$ eigenspace of $X_j$, $\langle\psi | Z_iZ_j|\psi\rangle=0$. It follows that for any state $|\chi\rangle$ satisfying $X_j|\chi\rangle=|\chi\rangle$, $\langle \chi | Z_iZ_j|\chi\rangle=0$.
 
The eigenspace of $Z_{i}Z_{j}$ is spanned by states of the form $|b_0b_1\dots b_N\rangle$ where  $b_i=b_j$. Acting with $X_j$ on this state flips $b_j$, and it follows that $\langle \psi | X_j |\psi\rangle=0$ for any $|\psi\rangle$ in the basis for the $1$ eigenspace of $Z_{i}Z_{j}$. Thus for any state $|\varphi\rangle$ satisfying $Z_iZ_j|\varphi\rangle=|\varphi\rangle$, $\langle \varphi |X_j|\varphi\rangle=0$.
\end{proof}
\section{The Transverse Ising Hamiltonian}
Given a graph $G$, the Transverse Ising Hamiltonian on graph $G$ can be defined as the following $N$-qubit Hamiltonian,
\begin{equation}
H=-J\sum_{\langle q_i,q_j\rangle \in E_G}Z_iZ_j-h\sum_{q_i \in V_G}X_i
\end{equation}
For our purposes we rescale this Hamiltonian as 
\begin{equation}
H=-\sum_{\langle q_i,q_j\rangle \in E_G}Z_iZ_j-g\sum_{q_i \in V_G}X_i
\end{equation}
for some non-negative dimensionless $g=\frac{h}{J}$. This rescaling has the effect of normalizing all eigenvalues by $J$. An overview of Transverse Ising models on $L_n$ and $P_n$ can be found in \cite{mbeng2020quantum} while an overview of Transverse Ising Models on random graphs can be found in \cite{IsingGraph}. A more general overview can be found in \cite{strecka2015brief}.
  
Minimizing the expectation value of $H$ acting on Clifford state $|\varphi\rangle$ then corresponds to minimizing the following,
\begin{equation}
    \langle \varphi | H | \varphi \rangle=-\sum_{\langle q_i,q_j\rangle \in E_G}\left\langle \varphi \right|Z_iZ_j\left|\varphi \right\rangle-g\sum_{q_i \in V_G}\langle \varphi |X_i|\varphi\rangle 
\end{equation}
Because all the Pauli operators above are Hermitian, every expectation value in the above expression is either $\pm 1$ or $0$. The following theorem defines a method for computing $\langle \varphi_0|H|\varphi_0\rangle$ by solving a related graph theoretic problem.

\begin{theorem}\label{Main Result}\normalfont
     Given \begin{equation} 
    V=\argmin_{S\in P(V_G)}-|E(S)|-g|V_G\backslash S|\end{equation} we can use $V$ to find $\varphi_0$ which satisfies
    \item \begin{equation} 
    \langle \varphi_0 | H | \varphi_0 \rangle = -|E(V)|-g|V_G\backslash V|\end{equation}
\end{theorem}
This theorem will follow from the proofs of Lemmas \ref{Part 1} and \ref{Part 2} below.
\begin{lemma}\label{Part 1}
Given an arbitrary vertex set $V\in P(V_G)$, we can always find a state $\varphi$ such that \begin{equation}
    \langle \varphi| H | \varphi\rangle = - |E(V)|-g|V_G\backslash V|\end{equation}
\end{lemma}
\begin{proof}
Let $S$ be the subgraph induced by $V$. By Lemma \ref{bfs}, there must be some spanning tree $M$ with $|E_M|<|S|$.
 
By definition, there cannot exist a sequence of edges of the form $\langle q_i,q_{k_1}\rangle,\langle q_{k_1},q_{k_2}\rangle,\dots,\langle q_{k_m},q_{j}\rangle $ if $\langle q_i,q_j\rangle\in E_M$ and vice versa. Thus it's possible to have a state $\varphi$ with $Z_iZ_j\in G_\varphi$ for all $\langle q_i,q_j\rangle E_M$ without violating Corollary \ref{allowed}.
 
If we include all these $Z_iZ_j$ terms, then we can also include $X_i\in G_\varphi$ if $q_i\notin V_M$ without violating Corollary \ref{allowed}.
 
Totalling the number of $Z_iZ_j$ and $X_i$ terms gives $|E_M|+|V_G\backslash S|$ many terms, which is clearly less than $N$. This suggests that we can compute a state $|\varphi\rangle$ such that 
\begin{itemize}
    \item $X_i\in G_\varphi$ if $q_i\notin V_M$
    \item $Z_iZ_j\in G_\varphi$ for all $\langle q_i,q_j\rangle \in E_M$
\end{itemize}
It can be verified via direct computation that the state 
\begin{equation}
    |\varphi\rangle  = \prod_{q_j\notin V}R_Y\left(\frac{\pi}{2}\right)_j|0\rangle
\end{equation}
Satisfies the desired operators being in $G_\varphi$.
 
Furthermore, this state $|\varphi\rangle$ satisfies the following (which can again be verified by direct computation).
\begin{itemize}
    \item $\langle \varphi | Z_iZ_j|\varphi\rangle=1$ if $\langle q_i,q_j\rangle \in E(V)$ 
    \item $\langle \varphi | Z_iZ_j|\varphi\rangle=0$ if $\langle q_i,q_j\rangle \notin E(V)$
    \item $\langle \varphi | X_i|\varphi\rangle=1$ if $ q_i \notin V$
    \item $\langle \varphi | X_i|\varphi\rangle=0$ if $ q_i \in V$
\end{itemize}
Thus, 
 $\langle \varphi| H | \varphi\rangle = - |E(V)|-g|V_G\backslash V|$ 
\end{proof}

\begin{lemma}\label{Part 2}\normalfont
Given a state $|\varphi\rangle$, we can always find a vertex set $V\in P(V_G)$ such that
\begin{equation}
\langle \varphi |H|\varphi \rangle \geq -|E(V)|-g|V_G\backslash V| \end{equation}
\end{lemma}
\begin{proof}
Let
\begin{equation}
    \mathcal{X}=\{X_0,X_1,\dots, X_{N-1}\}\end{equation}
And suppose that for some state $\varphi$, 
\begin{equation} S_\varphi \cap \mathcal{X}=\{X_{k_1},X_{k_2},\dots, X_{k_m}\}\end{equation}
Let $V=V_G\backslash\{q_{k_1},q_{k_2},\dots, q_{k_m}\}$. By Lemma \ref{ZX Sametime}, $\langle \varphi|Z_iZ_j|\varphi \rangle=0$ if $q_i,q_j\in \{q_{k_1},q_{k_2},\dots, q_{k_m}\}$.
 
It follows that the minimum possible value of $\langle \varphi | H | \varphi \rangle$ is clearly $-|E(V)|-g|V_G\backslash V|$. Thus, $V$ is the vertex set desired.
\end{proof}
The benefit of mapping the problem of finding $|\varphi_0\rangle$ to minimizing $-|E(S)|-g|V_G\backslash S|$ is that the latter is a \textit{submodular function} of $S$. It's well known that there exists a polynomial time algorithm that can be used to minimize these functions. Appendix \ref{comp} discusses algorithms which can be used to minimize such functions. 
 
We can also maximize over all subsets of a fixed size, which gives the following corollary relating this result to edge functions, 
\begin{corollary}\normalfont
The following statements follow from the proof of Theorem \ref{Main Result}
\begin{equation}
\langle \varphi | H | \varphi \rangle = \min_{n\in\{0,1,\dots N\}} -\mathcal{E}(n)-g(N-n)
\end{equation}
The vertex set minimizing $-|E(S)|-g|V_G\backslash S|$ will always be an $n$-optimal vertex set for some $0\leq n \leq N$.
\end{corollary}
While computing edge functions for specific values of $n$ is still NP-hard, maximizing functions of the above form is always possible in polynomial time.
 
Appendix \ref{comp} discusses algorithms connected to minimizing submodular functions and computing edge functions.

There are two regimes of $g$-values for which $H$ can be approximated as Hamiltonians for which the ground state is clearly Clifford.
\begin{itemize}
    \item If $g\to 0$, then \begin{equation}
    H\approx H_0=-\sum_{\langle q_i,q_j\rangle \in E_G}Z_iZ_j\end{equation} $|0\rangle$ is clearly a ground state of $H_0$.
    
    \item If $g\to \infty $, then \begin{equation}
    H\approx H_\infty=-g\sum_{q_i\in V_G}X_i\end{equation}
    $|+\rangle^{\otimes N}$ is clearly a ground state of $H_\infty$.
\end{itemize}
For convenience, let 
\begin{equation}
    C(n)=\mathcal{E}(n)+g(N-n)
\end{equation}
 
Since $N-n$ controls the number of $X_i$ terms in $G_\varphi$, we expect that for small $g$, $n=N$ will be optimal while for large $g$ we expect $n=0$ to be optimal.
 
The following Lemma verifies this behavior.
\begin{lemma}\label{extreme}\normalfont
\noindent
\begin{equation}
g\leq\min_{0\leq n< N}\frac{\mathcal{E}(N)-\mathcal{E}(n)}{N-n}\end{equation}
If any only if $N$ maximizes $C(n)$.
 
On the other hand,  
\begin{equation}g\geq \max_{0<n\leq N}\frac{\mathcal{E}(n)}{n}\end{equation}
If any only if $0$ maximizes $C(n)$.
\end{lemma}
\begin{proof}
Suppose that 
\begin{equation}
g\leq \min_{0\leq n< N}\frac{\mathcal{E}(N)-\mathcal{E}(n)}{N-n}
\end{equation}
This is equivalent to the following holding for all $0\leq n< N$
\begin{equation}
    g\leq \frac{\mathcal{E}(N)-\mathcal{E}(n)}{N-n}
\end{equation}
\begin{equation}
\mathcal{E}(N)\geq g(N-n)+\mathcal{E}(n)
\end{equation}
\begin{equation}
C(N)\geq C(n)
\end{equation}
Which is then equivalent to $n=N$ being optimal.
 
Now suppose that 
\begin{equation}
g\geq \max_{0<n\leq N}\frac{\mathcal{E}(n)}{n}
\end{equation}
This is equivalent to the following holding for all  $0<n\leq N$
\begin{equation}
g\geq\frac{\mathcal{E}(n)}{n}
\end{equation}
\begin{equation}
gN\geq g(N-n)+\mathcal{E}(n)
\end{equation}
\begin{equation}
C(0)\geq C(n)
\end{equation}
Which is then equivalent to, $n=0$ being optimal.
\end{proof}

This Lemma has the following corollary 
\begin{corollary}\label{twoseg}\normalfont
If $\frac{\mathcal{E}(N)}{N}\geq \frac{\mathcal{E}(n)}{n}$ for $0<n<N$, then 
\begin{equation}
\min_{0\leq n< N}\frac{\mathcal{E}(N)-\mathcal{E}(n)}{N-n} = \frac{\mathcal{E}(N)}{N}
\end{equation}
and furthermore,
\begin{equation}
\langle \varphi_0|H|\varphi_0\rangle=\begin{cases}
-\mathcal{E}(N) \textrm{\normalfont if } g\leq \frac{\mathcal{E}(N)}{N} \\ 
-gN\textrm{ \normalfont if } g\geq \frac{\mathcal{E}(N)}{N}
\end{cases}
\end{equation}
\end{corollary}
We call graphs for which $\frac{\mathcal{E}(N)}{N}\geq \frac{\mathcal{E}(n)}{n}$ for $0<n<N$ \textit{two-segmented} because of the result in Corollary \ref{twoseg}. For two segmented graphs, we call $\frac{\mathcal{E}(N)}{N}$ the transition value.
 
It turns out, computing the maximal value of $\frac{\mathcal{E}(n)}{n}$ alongside the maximal vertex set associated with this maximum is possible in polynomial time as discussed in Appendix \ref{comp}. 
\section{Numerical Experiments}
In this section we compare the optimal Clifford solution via our proposed method to ground states of Ising Hamiltonians obtained via exact diagonalization.

\begin{figure}
    \centering
    
    \includegraphics[width=\columnwidth]{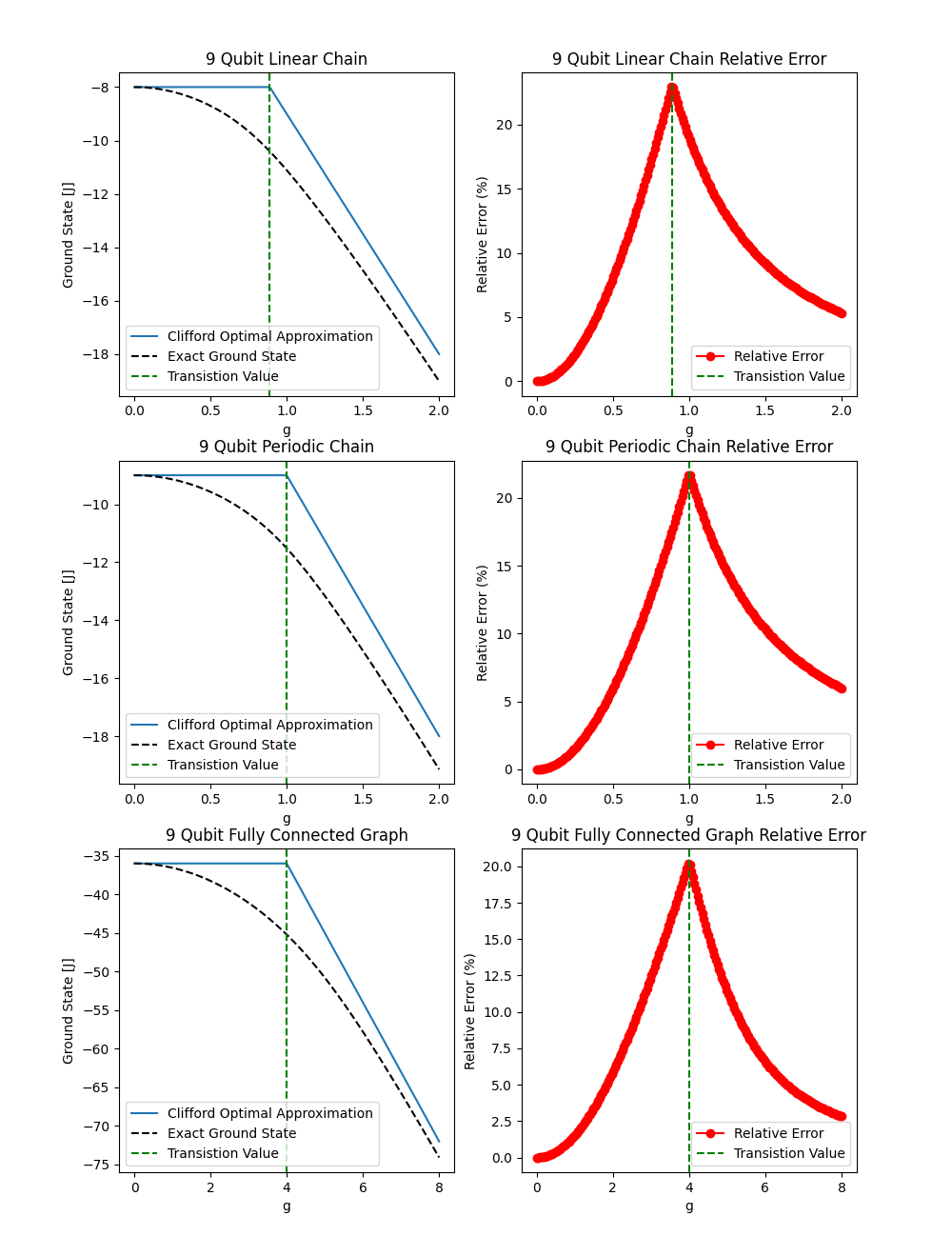}
    \caption{Comparison of the optimal Clifford state to the true ground state expectation for the selected $9$ qubit graphs. Notice that the maximal errors line up with the transition values and  that all of these graphs are two segmented.}
    \label{9LPK}
\end{figure}

\begin{figure}
    \centering
    \includegraphics[width=\columnwidth]{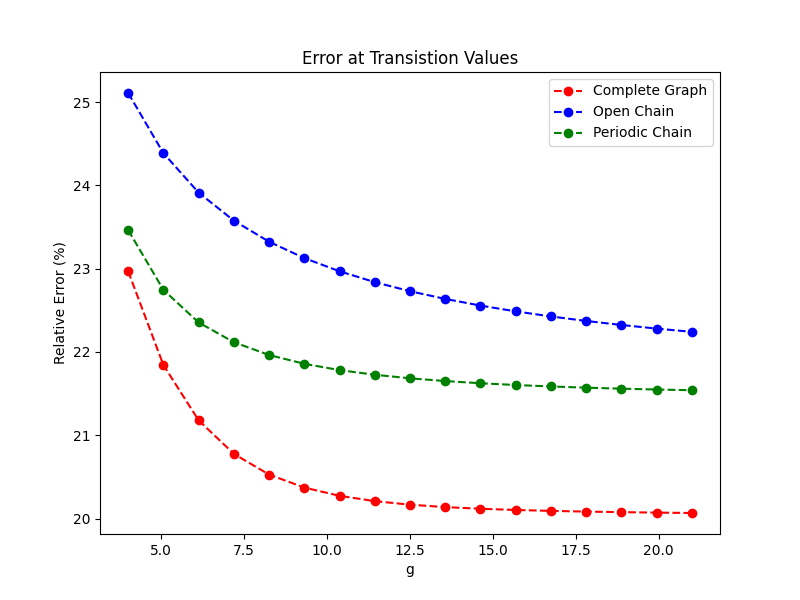}
    \caption{Plot of the relative error for $L_n,P_n,$ and $K_n$ for a range of values for $n$. Notice that all of the relative errors decrease with $n$ and eventually begin plateau.}
    \label{LPKError}
\end{figure}
\subsection{Linear Chains and Fully Connected Graphs}
Figure \ref{9LPK} plots a comparison between $\langle \varphi_0 | H | \varphi_0  \rangle$ and $\langle \psi_0 | H | \psi_0  \rangle$ for the Hamiltonians obtained from $L_9,P_9,$ and $K_9$ for different values of $g$.
 
Figure \ref{LPKError} plots the relative error for $L_n,P_n,K_n$ with $4\leq n \leq 20$ at the transition value of $g$.

\subsection{Other Graphs}
Figure \ref{Extra} plots the value of $\langle \psi_0 | H| \psi_0\rangle$ versus $\langle \varphi_0 | H| \varphi_0\rangle$ for $G_1,G_2,G_3$ for a range of $g$ values.
\subsection{Randomized Graphs}
Figure \ref{Randomized} plots the average relative error for random graphs. For $4\leq N\leq 17$, $100$ $N$ vertex graphs were generated by randomly placing a node between every pair of vertices with probability $0.5$. If the graph with no edges was obtained, this distribution was resampled.

\begin{figure*}[h]
    \centering
    \includegraphics[width=0.9\textwidth]{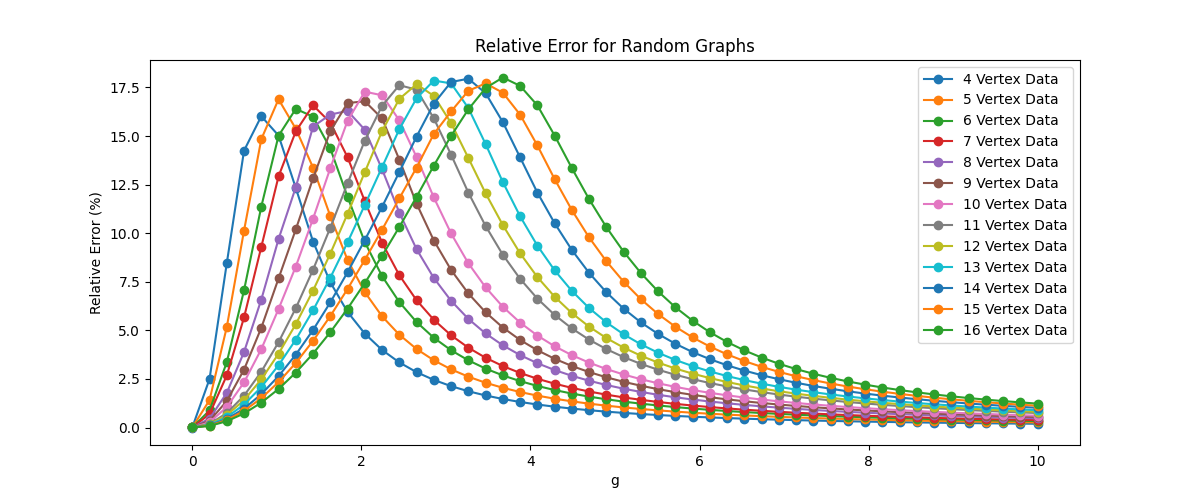}
    \caption{Plot of the mean relative error for different amounts of vertices. Notice that the value of $g$ for which the maximum mean relative error is obtained increases with $N$ as expected.}
    \label{Randomized}
\end{figure*}

\begin{figure}
    \centering
    \includegraphics[width=\columnwidth]{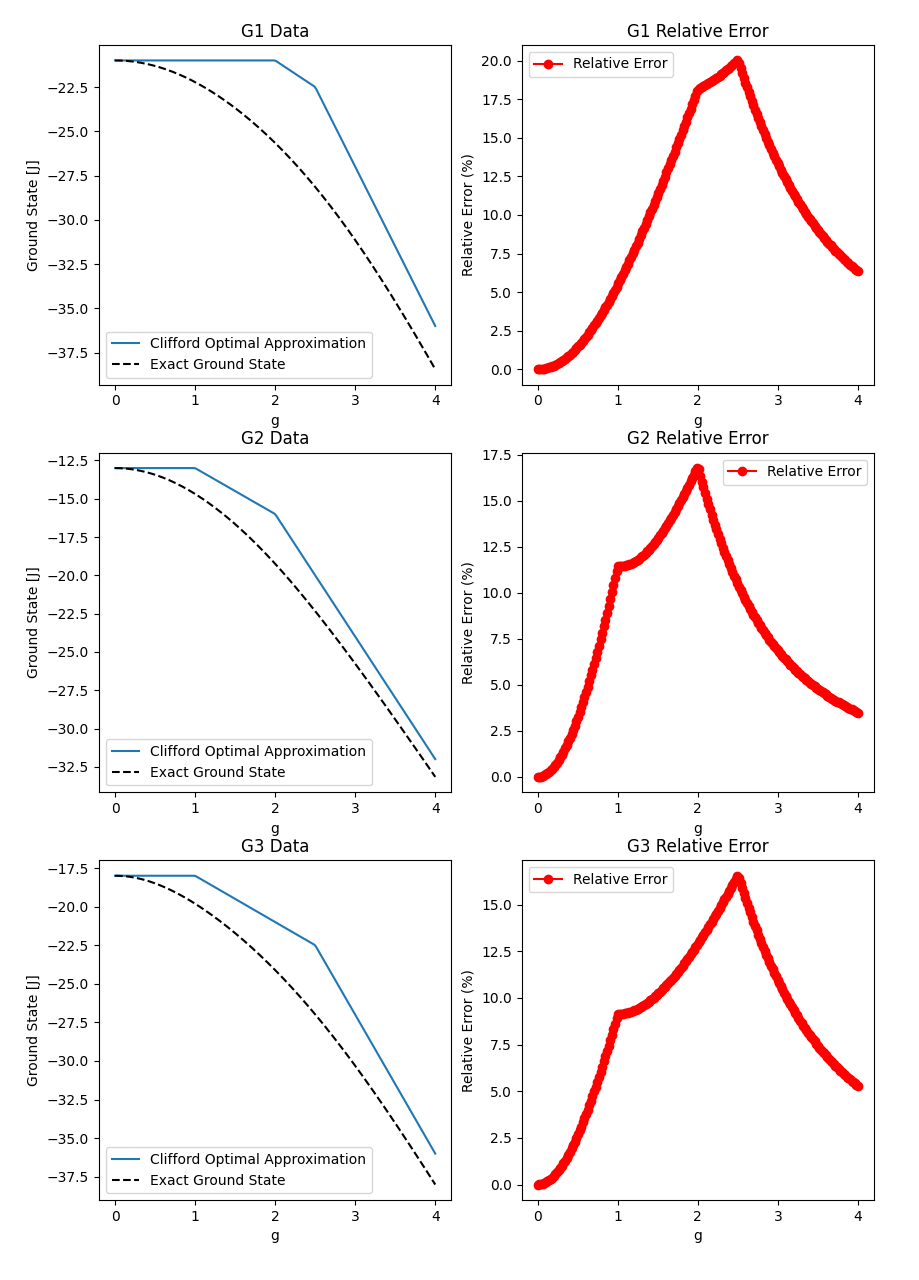}
    \caption{Comparison of the optimal Clifford state to the true ground state expectation for the selected graphs.  Notice that none of these graphs are two segmented.}
    \label{Extra}
\end{figure}

\section{Further Research}
\subsection{Other Hamiltonians}
Here we studied a specific spin-hamiltonian. There are multiple generalizations that could have been considered. 
\subsubsection{Weighted Ising Hamiltonians}
The most natural generalization of the Ising Hamiltonian considered here is the following
$$H=-\sum_{\langle q_i,q_j\rangle\in E_G}J_{ij}Z_iZ_j-\sum_{q_i\in V_G}J_iX_i$$
Assuming that $J_{ij}$ and $J_i$ are both positive for all values of $i,j$, we can once again tie the minimization of $\langle \varphi | H | \varphi \rangle$ to the minimization of a submodular function by incorporating these terms into the cost function.
\begin{corollary}
For a set of coefficients $J_{ij}$ and $J_i$,
$$\langle \varphi_0 | H | \varphi_0 \rangle = \min_{V\in  P(V_G)}-\sum_{\langle q_i,q_j\rangle \in E(V)}J_{ij}-\sum_{q_i \in V_G\backslash V}J_{i}$$
Where the function on the right side is a submodular function of $V$.
\end{corollary}
An example of an application of this modified Transverse Ising Hamiltonian, which accounts for different strength spin couplings can be found in \cite{Kaicher_2023}.  
In the case where $J_{ij},J_i$ are arbitrary, more careful analysis is required to handle negative coefficients because these can cause the corresponding graph theoretic function to no longer be submodular.
 
An interesting variation of this Hamiltonian is the version with the $X_i$ terms each replaced with $Z_i$. Such Hamiltonians often arise from QUBOs \cite{farhi2014quantum}\cite{blekos2023review}\cite{N_lein_2022}. These Hamiltonians are generally constructed so that all of their eigenvalues are Clifford, which generally implies that there is no non-NP-hard algorithm to find the optimal Clifford State for these Hamiltonians.
\subsubsection{Heisenberg Hamiltonians}
The general Heisenberg Hamiltonian can be defined as follows \cite{powell2010introduction}
$$H=-\sum_{\langle q_i,q_j\rangle}J_xX_iX_j+J_yY_iY_j+J_zZ_iZ_j$$
For arbitrary coefficients $J_x,J_y,J_z$. 
 
In the case where $J_x=J_y=J_z=1$, we obtain the $XXX$ model, and in the case where $J_x=J_y=1$ we obtain the $XXZ$ model.
 
Using methods similar to those in \cite{Rakov_2019}, it can be shown analytically that the ground state of the $XXX$ model and the $XXZ$ model with $J_z>1$ is always the Clifford state $|0\rangle$ regardless of the underlying graph $G$.
 
Finding the optimal Clifford state for a generalized Heisenberg Hamiltonian on an arbitrary graph will likely correspond to minimizing some constrained submodular function because of the additional restrictions on which terms in the Hamiltonian can simultaneously be in a given stabilizer. Unfortunately, this means that finding the optimal Clifford state for these problems may be NP-Hard  \cite{10.5555/3104482.3104605}.

\section{Conclusion}
In this work we outlined an efficient method for finding the optimal Clifford ground state of a given transverse field Ising Hamiltonian on an arbitrary graph. Our numerical estimates suggest that these optimal Clifford ground states are generally fairly accurate (with errors always less than $25\%$), suggesting that CAFQA-inspired approaches to VQE problems are useful initialization techniques for these Hamiltonians.

\section*{Acknowledgement}
This work is funded in part by EPiQC, an NSF Expedition in Computing, under award CCF-1730449; 
in part by STAQ under award NSF Phy-1818914; in part by NSF award 2110860; 
in part by the US Department of Energy Office  of Advanced Scientific Computing Research, Accelerated Research for Quantum Computing Program; 
and in part by the NSF Quantum Leap Challenge Institute for Hybrid Quantum Architectures and Networks (NSF Award 2016136) 
and in part based upon work supported by the U.S. Department of Energy, Office of Science, National Quantum Information Science Research Centers.  
This research used resources of the Oak Ridge Leadership Computing Facility, which is a DOE Office of Science User Facility supported under Contract DE-AC05-00OR22725.

\begin{appendices}

\section{Graph List}
\label{graph}
The following graphs are used throughout for computations. Notice every graph here is connected.
\subsection{Linear Chains}\label{lineargraphs}
The $n$-vertex open linear chains, $L_n$, with $n\geq 2$ will be defined as the graphs with nodes $\{q_0,q_1,\dots, q_{n-1}\}$
and only edges from $q_i$ to $q_{i+1}$ for $1\leq i < N$. Figure \ref{Open linear chains} shows the first $5$ open linear chains.
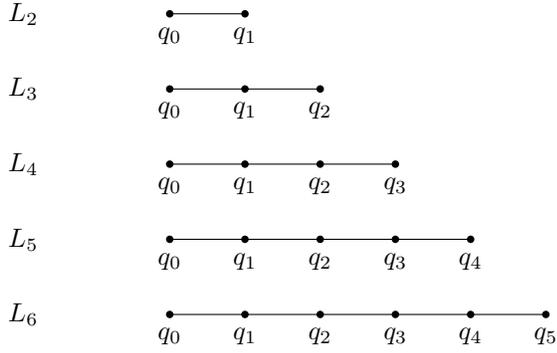
\begin{figure}
    \centering
    \begin{tikzpicture}
     \foreach \n in {1,...,5}{
     
      \foreach \m in {0,...,\n}{\draw (\m+1,-\n-1)node[point,label={below:{$q_{\m}$}}]{};}
     
     }
    \draw (1,-2)--(2,-2) (1,-3)--(3,-3) (1,-4)--(4,-4) (1,-5)--(5,-5) (1,-6)--(6,-6);
    
    \foreach \n in {2,...,6}{\draw (0,{-1*\n})node[label={[label distance=0.5cm]180:{${L}_{\n}$}}]{};}
    \end{tikzpicture}
    \caption{First $5$ open linear chains}
    \label{Open linear chains}
\end{figure} 
The $n$-vertex periodic linear chains, ${P}_n$, with $n\geq 3$ will be defined as ${L}_n$ with the additional edge from $q_{n-1}$ to $q_1$. The first $5$ periodic chains are drawn in Figure  \ref{Periodic linear chains}.

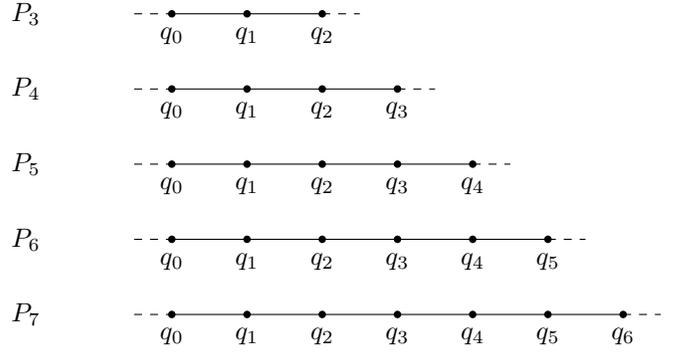
\begin{figure}
    \centering
    \begin{tikzpicture}
     \foreach \n in {2,...,6}{
     
      \foreach \m in {0,...,\n}{\draw (\m+1,-\n-1)node[point,label={below:{$q_\m$}}]{};}
     
     }
    \draw (1,-7)--(7,-7) (1,-3)--(3,-3) (1,-4)--(4,-4) (1,-5)--(5,-5) (1,-6)--(6,-6);
    \foreach \n in {3,...,7}{\draw (0,{-1*\n})node[label={[label distance=0.5cm]180:{${P}_{\n}$}}]{};}
       \foreach \n in {3,...,7}{\draw[dashed] (0.5,{-1*\n})--(1,-1*\n)
       (\n,{-1*\n})--(\n+0.5,-1*\n);}
    \end{tikzpicture}
    \caption{First $5$ periodic linear chains}
    \label{Periodic linear chains}
\end{figure}
\subsection{Fully Connected Graphs}\label{fullgraphs}
The fully connected graph with ${K}_4$ with $n\geq 4$ will be defined as the graph with vertices $\{q_0,q_1,\dots,q_{n-1}\}$ and an edge between every pair of vertices. Figure \ref{Fully connected graphs} shows the first $3$ fully connected graphs.
\begin{figure*}
    \centering
    \begin{tikzpicture}[scale=0.8]
     \foreach \n in {3,...,5}{
     
      \foreach \m in {0,...,\n}{\draw ({cos((\m+1)*360/(\n+1))+6*\n},{sin((\m+1)*360/(\n+1))})node[point,label={360*(\m+1)/(\n+1):{$q_{\m}$}}]{};
      
      }
     \foreach \x in {0,...,\n}{\foreach \y in {0,...,\n}
     {
     \draw ({cos(\x*360/(\n+1))+6*\n},{sin(\x*360/(\n+1))})--({cos(\y*360/(\n+1))+6*\n},{sin(\y*360/(\n+1))});
     }
     }
     
    }
    \foreach \n in {4,...,6}{\draw({6*(\n-1)},-3)node[label={${K}_{\n}$}]{};}
    
    \end{tikzpicture}
    \caption{First $3$ fully connected graphs}
    \label{Fully connected graphs}
\end{figure*}
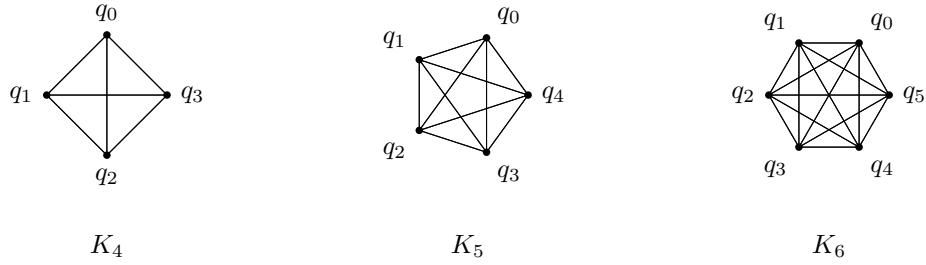

\subsection{Miscellaneous Graphs}\label{miscgraph}
The following graphs (Figure \ref{extras}) are designed to have subgraphs which are more dense than the entire graph, and thus are all not two-segmented.
\begin{figure*}
    \centering
    \begin{tikzpicture}[scale=0.8]
   \foreach \n in {0,...,5}{
    \draw({0.5+cos(60*\n)},{0.5+sin(60*\n)})node[point]{}--({0.5+cos(60*\n+60)},{0.5+sin(60*\n+60)});
    \draw({0.5+cos(60*\n)},{0.5+sin(60*\n)})--({0.5+cos(60*\n+120)},{0.5+sin(60*\n+120)});
   
    \draw({0.5+cos(60*\n)},{0.5+sin(60*\n)})--({0.5+cos(60*\n+180)},{0.5+sin(60*\n+180)});
    }
    \draw(-1,-0.366)node[point,label=below:{$q_0$}]{}--(2,-0.366)node[point,label=below:{$q_3$}]{}--(0.5,2.232)node[point,label=above:{$q_6$}]{}--(-1,-0.366);
    \draw(0,-0.366)node[point,label=below:{$q_1$}]{};
    \draw(1,-0.366)node[point,label=below:{$q_2$}]{};
 \draw(1.5,0.5)node[point,label=right:{$q_4$}]{};
 \draw(1,0.5+0.866)node[point,label=right:{$q_5$}]{};
  \draw(-0.5,0.5)node[point,label=left:{$q_8$}]{};
 \draw(0,0.5+0.866)node[point,label=left:{$q_7$}]{};

    \foreach \n in {1,...,4}{
    \draw({5+cos(72*\n)},{0.5+sin(72*\n)})node[point,label={72*\n}:$q_\n$]{}--({5+cos(72*\n+72)},{0.5+sin(72*\n+72)});
    \draw({5+cos(72*\n)},{0.5+sin(72*\n)})--({5+cos(72*\n +144)},{0.5+sin(72*\n+144)});
   
    }
    \draw(6,0.5)node[point,label=90:$q_0$]{}--({5+cos(72)},{0.5+sin(72)});
    \draw(6,0.5)--({5+cos(144)},{0.5+sin(144)});

    \draw(6,0.5)--(7,0.5)node[point,label=above:$q_5$]{}--(8,1.5)node[point,label=above:$q_6$]{} (7,0.5)--(8,-0.5)node[point,label=below:$q_7$]{};

    \foreach \n in {0,...,5}{
    \draw({12+cos(60*\n)},{0.5+sin(60*\n)})node[point]{}--({12+cos(60*\n+60)},{0.5+sin(60*\n+60)});
    \draw({12+cos(60*\n)},{0.5+sin(60*\n)})--({12+cos(60*\n+120)},{0.5+sin(60*\n+120)});

   \draw(13.2,0.2)node[label={above:$q_0$}]{};
      \draw(12.5,1.2)node[label={above:$q_1$}]{};
      \draw(11,1)node[label={above:$q_2$}]{};
   \draw(11.3,0.5)node[label={left:$q_3$}]{};
   \draw(11.6,-0.3)node[label={below:$q_4$}]{};
      \draw(12.5,-.7)node[label={right:$q_5$}]{};

    \draw({12+cos(60*\n)},{0.5+sin(60*\n)})--({12+cos(60*\n+180)},{0.5+sin(60*\n+180)});
    }
    \draw({12+cos(0)},{0.5+sin(0)})--({12+2*cos(0)},{0.5+2*sin(0)})node[point,label=0:$q_6$]{};
    
    \draw({12+cos(120)},{0.5+sin(120)})--({12+2*cos(120)},{0.5+2*sin(120)})node[point,label=120:$q_7$]{};
    
    \draw({12+cos(240)},{0.5+sin(240)})--({12+2*cos(240)},{0.5+2*sin(240)})node[point,label=240:$q_8$]{};
    \draw (12,-1.5)node[label=below:$G_3$]{};
    \draw (5,-1)node[label=below:$G_2$]{};
    \draw (0.5,-1.5)node[label=below:$G_1$]{};
    \end{tikzpicture}
    \caption{Various non two-segmented graphs.}
    \label{extras}
\end{figure*}
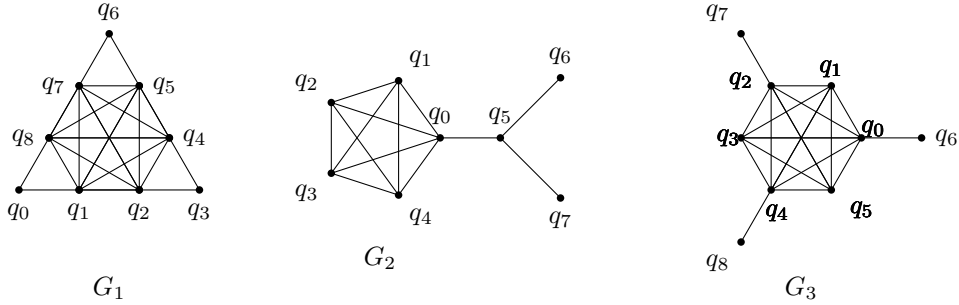
These graphs will be referred as $G_1,G_2,$ and $G_3$.
\section{Overview of relevant algorithms}\label{comp}
In this appendix we provide a brief overview of the relevant algorithms used in this paper. All relevant code can be found at \cite{Bhattacharyya_OptimalCliffordIsing}.
\subsection{Edge Functions (DkS)}
The \textit{densest k-subgraph} (DkS) problem can be defined as the problem of finding the vertex set $V\subset V_G$ with $|V|=k$ for fixed $k$ such that $\frac{|E(V)|}{|V|}$ is maximized. This algorithm is known to be NP-hard \cite{lanciano2023survey}.
 
Clearly, computing $\mathcal{E}(n)$ for $n\in \{0,1,\dots, N\}$ is equivalent to solving DkS for $k\in\{0,1,\dots, N\} $, meaning that computing all of the values for $\mathcal{E}(n)$ is an NP-hard problem.
 
Currently, the best known approximate algorithm for DkS was introduced in \cite{bhaskara2010detecting}, and computes an $O(n^{1/4+\epsilon})$ approximation in $O(n^{1/\epsilon})$.
 
Although edge functions are useful theoretical tools, we generally don't need to consider specific values of $\mathcal{E}$ when searching for the optimal Clifford state, meaning we can avoid this NP-Hard problem.
\subsection{Two-Segmented Testing (DSP)}
The \textit{densest subgraph problem} (DSP) can be defined as the problem of finding the vertex set $V\subset V_G$ such that $\frac{|E(V)|}{|V|}$ is maximized. It is well known that the Densest Subgraph Problem can be solved in polynomial time \cite{lanciano2023survey}.   
 
One solution method uses a linear programming problem in $N+|E_G|$ variables, which was introduced in \cite{10.5555/646688.702972}. Another polynomial time solution was introduced in \cite{Goldberg:CSD-84-171} and uses maximum-flow computations.
 
Starting with an arbitrary graph $G$, we can always run a solution to DSP on it, and if the obtained vertex set $V$ satisfies $\frac{|E(V)|}{|V|}=\frac{|E(N)|}{|N|}$, we immediately know that $G$ must be two segmented.
\subsection{Submodular Minimization}\label{sub}
The most important algorithm used here is the algorithm to minimize the function 
$$f(V)  = -|E(V)|-g|V_G\backslash V|$$
The following corollary immediately follows from the definition $f$,
\begin{corollary}\label{submodular}
For any sets $A,B\subset V_G$,
$$f(A)+f(B)\leq f(A\cap B) + f(A\cup B) $$
\end{corollary}
\vspace{1cm}
Functions $f$ which satisfy corollary \ref{submodular} are called \textit{submodular} functions. A comprehensive overview of submodular functions can be found in \cite{subBook}.
 
One of the most important properties of submodular functions is that they can always be minimized in polynomial time. The first algorithm to do so was introduced in \cite{iwata2000combinatorial} and \cite{SCHRIJVER2000346}, but this algorithm is seemingly too slow for practical applications. Another independent algorithm was introduced in \cite{NewModular}, which uses an algorithm to minimize the norm of a point in a polyhedra introduced in \cite{Wolfe1976FindingTN} and theory developed in \cite{10.5555/885909.885912}. An approximation for the performance of this algorithm was obtained in \cite{chakrabarty2014provable}, which provided a pseudopoynomial bound on the runtime of this algorithm.
\end{appendices}

\bibliographystyle{IEEEtran} 
\bibliography{references}
\end{document}